\documentclass[aps,prl,twocolumn,showpacs,showkeys,a4]{revtex4-1}
\usepackage{graphicx}
\usepackage{amsmath}
\usepackage{amssymb}

\usepackage[english]{babel}
\usepackage{dsfont}
\usepackage{bbm}
\usepackage{hyperref} 
\usepackage{amsmath, amsthm, amssymb, mathrsfs}
\usepackage{times}

\usepackage[pdftex,usenames]{color}

\newtheorem{theorem}{Theorem}

\DeclareMathOperator{\Tr}{\mathrm{Tr}}%trace2
\DeclareMathOperator{\tracedistance}{\mathcal{D}}
\DeclareMathOperator{\ee}{\mathrm{e}}%Euler e
\DeclareMathOperator{\iu}{\mathrm{i}}%imaginary unit
\DeclareMathOperator{\hiH}{\mathcal{H}}%Hilbert spaces
\DeclareMathOperator{\haH}{\mathscr{H}}%Hamiltonians
\DeclareMathOperator{\id}{\mathds{1}}%identity matrix
\newcommand{\deff}{d^{\mathrm{eff}}}%effective dimension

\newcommand{\bra}[1]{\langle #1|}
\newcommand{\ket}[1]{|#1\rangle}
\newcommand{\braket}[2]{\langle #1|#2\rangle}
\newcommand{\ketbra}[2]{| #1 \rangle \langle #2 |}

\newcommand{\basisentanglement}{R}
\newcommand{\eqcoef}{\ensuremath{{C_{\mathrm{eq}}}}}

\begin{document}
\title{Absence of Thermalization in Nonintegrable Systems}
\author{Christian Gogolin$^{1,2,3}$, Markus P.\ M{\"u}ller$^{1,4}$, and Jens Eisert$^{1,5}$}

\affiliation{$^1$ Institute for Physics and Astronomy, Potsdam University, 14476 Potsdam, Germany\\
$^2$ Fakult{\"{a}}t f\"{u}r Physik und Astronomie, Universit\"{a}t W\"{u}rzburg, Am Hubland, 97074 W\"{u}rzburg, Germany\\
$^3$ Department of Mathematics, University of Bristol, University Walk, Bristol BS8 1TW, United Kingdom\\
$^4$ Institute of Mathematics, Technical University of Berlin, 10623 Berlin, Germany\\
$^5$ Institute for Advanced Study Berlin, 14193 Berlin, Germany}

%\author{Christian Gogolin}
%\affiliation{\bristol}
%\affiliation{\potsdam}
%\affiliation{\wuerzburg}
%\author{Markus M\"{u}ller}
%\affiliation{\berlin}
%\affiliation{\potsdam}
%\author{Jens Eisert}
%\affiliation{\wiko}
%\affiliation{\potsdam}

\begin{abstract}
We establish a link between unitary relaxation dynamics after a quench in closed many-body systems and the entanglement in the energy eigenbasis.
We find that even if reduced states equilibrate, they can have memory on the initial conditions even in certain models that are far from integrable.
We show that in such situations the equilibrium states are still described by a maximum entropy or generalized Gibbs ensemble, regardless of whether a model is integrable or not, thereby contributing to a recent debate.
In addition, we discuss individual aspects of the thermalization process, comment on the role of Anderson localization, and collect and compare different notions of integrability.
\end{abstract}

% Explanation of PACS numbers:
% 03.65.-w      Quantum mechanics
% 05.30.-d      Quantum statistical mechanics
% 03.65.Yz      Decoherence; open systems; quantum statistical methods
% 05.70.Ln      Irreversible thermodynamics

\maketitle

The question of how quantum many-body systems in nonequilibrium eventually equilibrate and assume properties resembling the ones familiar from statistical mechanics has---quite unsurprisingly---a very long tradition \cite{SchroedingerNeumann}.
In closed systems not all observables can equilibrate.
However, it is generally expected that in sufficiently complicated quantum many-body systems at least some physically relevant quantities should seemingly relax to equilibrium values.
Recently this old question has received an enormous amount of attention and there have been significant new insights \cite{EigenstateThermalization,Cramer,PRE81,Linden09,Popescu06,Reimann08,Pal10,Goldstein06,%Deutsch10,
tasaki98,PRE82,Gogolin10,Gemmer09,Cazalilla10,Barthel,Calabrese,Rigol07,Rigol08}.

This renewed attention is partly driven by new mathematical methods becoming available, partly by novel numerical techniques, and in parts by experiments that make it possible to probe coherent nonequilibrium dynamics under the controlled conditions offered by cold atoms in optical lattices \cite{Bloch}. 
Theoretically, among other approaches, the question of how quantum many-body systems relax locally has been investigated in the light of the  ``eigenstate thermalization hypothesis'' (ETH) \cite{EigenstateThermalization,Rigol08}, quantum central limit theorems \cite{Cramer}, Anderson localization \cite{Pal10}, dynamical instances of concentration of measure arguments or ideas of relaxation via dephasing \cite{Linden09,Popescu06,Reimann08,PRE82,PRE81,Gogolin10,Gemmer09}, and numerically using time-dependent density-matrix renormalization group (DMRG) \cite{kollschollmanm}.
Despite this enormous effort, major questions remain open and the existing results do not yet draw a coherent picture.
What seems to have become consensus, however,  is that the following expectation holds true: Nonintegrable systems thermalize.

In this letter we show that generally, this is not quite true.
We do so by establishing a link between the \emph{entanglement in the eigenbasis} of a quantum many-body system with what could be called the \emph{thermalization potential} of the system.
We will investigate situations in which systems \emph{equilibrate}, in the sense that all local observables will be close to some equilibrium value at most times, but those values turn out to depend on the details of the initial state.
This general rigorous statement is exemplified numerically by studying a small natural nonintegrable XYZ-type spin chain model.
In previous approaches (e.g., in Ref.~\cite{BanulsCiracHastings10}), similar complementing observations have been made by simulating the model's time evolution explicitly.
However, such simulations can only trace the system's behavior for a finite amount of time and become unreliable for long times.
Our analytic results have applications far beyond this particular model: they yield general conditions for the absence of thermalization.
This gives a natural counterpart of the ETH, and it relates the thermalization of isolated quantum systems to the presence of entanglement in the energy eigenbasis.

\begin{figure}[t]
 \centering
 \includegraphics[width=0.9\linewidth]{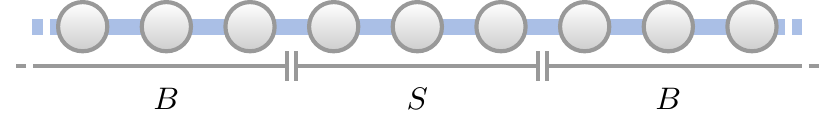}
 \caption{(color online)
   A many-body quantum system consisting of a small subsystem $S$, and the remainder $B$, which acts as a ``bath.''
 }
 \label{fig:setup}
\end{figure}

\paragraph{Setup and notation.} 
When using terms and concepts borrowed from classical statistical mechanics, such as ergodicity, equilibration, thermalization, initial state independence, and integrability, we aim at being careful and precise whenever referring to one of these terms.
We work in the pure state quantum statistical mechanics model with a system and bath setup with a global pure state and unitary time evolution \cite{Cramer,Linden09,Popescu06,Gogolin10,PRE81,PRE82}.
We are mostly interested in the case where the full system is composed of many interacting small systems and the subsystem corresponds to a small subset of sites and the bath is simply the remainder.% (see Fig.~\ref{fig:setup}).
In such systems the individual subsystems act as baths for each other and the collective dynamics can lead to self-thermalization of the whole system.
To be specific, we will consider arbitrary quantum systems equipped with a Hilbert space $\hiH$ of finite dimension $d$ that can be divided into at least two parts, i.e., $\hiH = \hiH_S \otimes \hiH_B$, which we will call the subsystem $S$ and the bath $B$, and which are described by Hilbert spaces of dimensions $d_{S/B} = \dim(\hiH_{S/B})$.
We assume that at every time $t$ the joint system is in a pure state $\psi_t = \ketbra{\psi_t}{\psi_t}$, evolving unitarily.
The reduced states on subsystem and bath are denoted using superscript letters, such as $\psi^S =\Tr_B[\psi]$.% and $\psi^B =\Tr_S[\psi]$.
We denote the Hamiltonian of the full system by $\haH$ and its eigenvectors and eigenvalues by $\ket{E_k}$ and $E_k$, $k=1,\dots,d$ (regardless of degeneracies).

\paragraph{Thermalization.}
Thermalization is a complicated process.
For it to happen a system must exhibit certain properties, each of which captures a specific aspect of thermalization.
It is instructive to consider each of them separately \cite{Linden09,PRE82}:

1) {\it Equilibration:}
The tendency to evolve towards equilibrium is a key assumption in classical statistical physics and part of the second law of thermodynamics.
In contrast to that, in the framework of pure state quantum statistical mechanics, equilibration for almost all times can instead be proven to be a consequence of unitary time evolution \cite{Linden09}.

The remaining conditions specify what properties the equilibrium state of a small subsystem should have and of which of the initial conditions it should be independent.

2) {\it Subsystem initial state independence:}
The equilibrium state of a small subsystem should be independent of the initial state of that subsystem.
This aspect of thermalization will be the main subject of the present work.

3) {\it Diagonal form of the subsystem equilibrium state:}
The equilibrium state of a small subsystem should be (close to) diagonal in the energy eigenbasis of its self-Hamiltonian \cite{PRE81}.

4) {\it Bath state independence:}
It is generally expected that the equilibrium expectation values of local observables on a small subsystem are almost independent of the details of the initial state of the rest of the system, but rather only depend on its macroscopic properties, such as the energy density.

5) {\it Gibbs state:}
Ultimately, one would like to recover the standard assumption of classical statistical physics that the equilibrium state is of (or at least close to) a Gibbs state $\omega^S \approx \ee^{-\beta\,\haH_S}$ with an inverse temperature $\beta$ and a self Hamiltonian $\haH_S$.
As in conventional statistical physics this can only be expected to be true if the coupling is weak but nonperturbative and the bath has a spectrum that gets exponentially dense for higher energies \cite{tasaki98,Goldstein06}.

\paragraph{Integrability.}
In classical mechanics integrability is a well-defined concept \cite{CI}. In quantum 
mechanics, despite the common use of the term ``integrable,'' the situation is much less clear \cite{Caux10}, and different criteria are being applied in the literature. The most common notions of integrability are the following:

(A) There exist $n$ independent (local) \cite{CI2} 
conserved mutually commuting linearly independent operators, where $n$ is the number of degrees of freedom (see, e.g., Ref.~\cite{integrable}).
In contrast to the classical situation \cite{CI}, this does not necessarily imply that the system is ``exactly solvable.''

(B) Identical with criterion (A), but with linear independence replaced by algebraic independence \cite{integrable}.

(C) The system is integrable by the Bethe ansatz \cite{integrable}.

(D) The system exhibits nondiffractive scattering \cite{integrable}.

(E) The quantum many-body system is exactly solvable in any way. 
Of course, this criterion is subject to the ambiguity of a lack of imagination of solving a given model.

\paragraph{Equilibration.}
Quantum mechanics of closed systems is time reversal invariant and thus equilibration in the usual sense is impossible.
Therefore we use an extended notion of equilibration and say that an observable $A$ equilibrates if its expectation value $\Tr[A\,\psi_t]$ is close to some value for almost all times $t$.
Of particular interest are local observables, i.e., observables that are sums of terms that each act only on small subsystems.
Saying that all observables on some subsystem $S$ equilibrate is equivalent to saying that the state $\psi^S_t$ of the subsystem equilibrates, by which we mean that there exists a state $\rho^S$ such that $\psi^S_t$ is almost always physically indistinguishable from $\rho^S$, in the sense that their trace distance $\tracedistance(\psi_t^S,\rho^S) = \frac{1}{2}\| \psi_t^S -\rho^S\|_1= \max_{0\leq A \leq\mathbb{I}} \Tr[A\,\psi^S_t] - \Tr[A\,\rho^S]$ is small for almost all times $t$.
If the expectation value $\Tr[A\,\psi_t]$ of an observable $A$ equilibrates in the sense defined above, then it must equilibrate towards its time average $\overline{\Tr[A\,\psi_t]} = \Tr[A\,\overline{\psi_t}]$.
This is an obvious but important observation.
A good understanding of the properties of the time averaged state
\begin{equation}
  \omega = \overline{\psi_t} = \lim_{\tau\to\infty} \frac{1}{\tau} \int_0^\tau \psi_t\,dt 
\end{equation}
is thus key to understanding equilibrium properties.

\paragraph{Generalized Gibbs ensemble.}
Every Hamiltonian $\haH$ defines a set of conserved observables. 
In the nondegenerate case they are exactly the linear combinations of projectors onto the eigenstates of $\haH$, 
in the degenerate case they are the observables with support on the blocks corresponding to the degenerate subspaces.
Clearly, the time average $\omega=\overline{\psi_t}$ of the state $\psi_0$  itself is given by $\omega=P(\psi_0)$, where $P(\psi_0)= \sum_j \pi_j \psi_0 \pi_j$, where $\pi_j= \sum_{k\in I_j} \ket{E_k} \bra{E_k}$ are the projections onto (possibly degenerate) eigenspaces, $E_k = E_{l}$ for $k,l\in I_j$. 
Every state $\rho$ that gives the same values for all conserved observables as $\psi_0$ satisfies $P(\psi_0) = P(\rho)$ and $\omega= P(\psi_0)$ is the state having {\it maximum entropy} among all such states.
This follows directly from the pinching inequality (Theorem~V.2.1 in Ref.~\cite{bhatia}) since the von Neumann entropy is Schur concave.
All equilibrium expectation values can be calculated from the maximum entropy state $\omega$.
This is a {\it quantum version of Jaynes'-principle} and was recently conjectured as {\it generalized Gibbs ensemble} in Ref.~\cite{Rigol07}.

Moreover, under the assumption of nondegenerate energy gaps it can be rigorously proven under which conditions equilibration (but not necessarily thermalization) happens \cite{Reimann08,Linden09}.
The certificate quantifying the quality of equilibration is the {\it effective dimension} of the time averaged state $\deff(\omega) = 1/\Tr[\omega^2]$, which, for quenches to nondegenerate Hamiltonians, is identical to the inverse of the time average of the Loschmidt echo and the inverse participation ratio (IPR) of the initial state \cite{epaps}.
The main result of Ref.~\cite{Linden09} is
\begin{equation}
  \label{eq:equilibrationtheorem}
  \overline{\tracedistance(\psi^S_t,\omega^S)} \leq \frac{1}{2} \sqrt{\frac{d_S^2}{\deff(\omega)}} = \eqcoef(\psi_0) ,
\end{equation}
and we call $\eqcoef(\psi_0)$ the {\it equilibration coefficient} of the initial state $\psi_0$ as it bounds the trace norm equilibration radius.

\paragraph{Main result.}
In systems that behave thermodynamically the equilibrium expectation values of local observables on small subsystems should be independent of the initial state of the subsystem.
A previous positive result in this direction was made in Ref.~\cite{Linden09} (see also Refs.~\cite{Gogolin10,PRE82}). 
Here we follow a converse approach and proof a sufficient condition for the \emph{absence of initial state independence}. 

A quantity that will play an important role in our main result is the \emph{effective entanglement in the eigenbasis}, given for a nondegenerate $\haH$ by
\begin{equation}
  \label{eq:definitionofourentaglementmeasure}
  \basisentanglement(\psi_0) = \sum_k |c_k|^2 \tracedistance(\Tr_B\ketbra{E_k}{E_k},\psi^S_0) ,
\end{equation}
with $c_k = \braket{E_k}{\psi_0}$.  
This quantity is small, if most energy eigenstates either resemble locally the system's initial state $\psi^S_0$, or are globally almost orthogonal to $ \psi_0$.
As will become apparent later, this is in particular the case if the reductions of the $\ketbra{E_k}{E_k}$ are close to a basis for $S$.
This can be interpreted as a natural counterpart of the ETH \cite{EigenstateThermalization}: If ``most'' energy eigenstates have reduced states close to some $\rho^S$, then the system will relax locally to $\rho^S$.

We will now show that a small value of $R$ implies that initial state independence is not satisfied.
Remarkably, this is not a matter of time scales:
It will not only take a long time to relax; but one will rather encounter a memory for almost all times.
\begin{theorem}[Nonthermalization]
  \label{theo:mainresult}
  The physical distinguishability of the two local time averaged states $\omega^{S(1)}$ and $\omega^{S(2)}$ of two pure initial product states $\psi^{(i)}_0 = {\psi^{S(i)}_0} \otimes {\phi^{B(i)}_0} ,\quad i \in \{1,2\}$ evolving under a nondegenerate Hamiltonian $\haH$ is large in the sense that
  \begin{equation*}
    \tracedistance(\omega^{S(1)},\omega^{S(2)}) \geq \tracedistance(\psi^{S(1)}_0,\psi^{S(2)}_0) -  \basisentanglement(\psi_0^{(1)}) - \basisentanglement(\psi_0^{(2)}) .
  \end{equation*}
  In the degenerate case, the quantity $R$ has to be replaced by
\begin{equation*}
  \basisentanglement(\psi_0)= \sum_k \langle \psi_0| \pi_k |\psi_0\rangle
  \tracedistance\biggl(
  \frac{\Tr_B(\pi_k \psi_0 \pi_k)}{ \langle \psi_0| \pi_k |\psi_0\rangle},\psi^S_0\biggr).
\end{equation*}
\end{theorem}
That is to say, subsystems remain distinguishable if they are initially well distinguishable and one has little effective entanglement in the eigenbasis.
Note also that the environment states ${\phi^{B(1)}_0}$ and ${\phi^{B(2)}_0}$ can be taken to be identical.
\begin{proof} 
  If the Hamiltonian $\haH$ is nondegenerate, $ \omega^{S(i)} = \sum_k |c^{(i)}_k|^2 \Tr_B \ketbra{E_k}{E_k}$, and thus
   \begin{align*}
    \tracedistance(\psi^S_0,\omega^S) &= \frac{1}{2} \bigl\| \psi^S_0 - \sum_k |c_k|^2 \Tr_B \ketbra{E_k}{E_k} \bigr\|_1 \\
    &\leq \sum_k |c_k|^2 \frac{1}{2} \left\|  \psi^S_0 - \Tr_B \ketbra{E_k}{E_k} \right\|_1= R(\psi_0) .
  \end{align*}
 The desired result then follows from $ \tracedistance(\psi^{S(1)}_0,\psi^{S(2)}_0) \leq \tracedistance(\psi^{S(1)}_0,\omega^{S(1)})+  \tracedistance(\omega^{S(1)},\omega^{S(2)}) + \tracedistance(\omega^{S(2)},\psi^{S(2)}_0)$. In the degenerate case, the same argument can be followed for the projectors $\pi_k$ onto the respective eigenspaces.
\end{proof}
The intuition that $R$ will be small when the $\ket{E_k}$ have certain properties can be made rigorous in the case where the $\Tr_B \ketbra{E_k}{E_k}$ are close to a basis of $S$.
In this case we can show that there exist many initial states for the bath that lead to a small $R$ and a large effective dimension at the same time, thus causing ``equilibration without thermalization.''
This is shown using Haar-measure averages, from which the existence follows \cite{Remark2}.

\begin{theorem}[Entanglement in eigenbasis]
  \label{theo:bounds}
  For every orthonormal basis $\{\ket{i}\}$ for $S$ and every initial product state with $\psi^S_0 = \ketbra{i}{i}$ for some $i$ and with Haar random initial bath part $\phi^B_0$, the effective entanglement in the eigenbasis for nondegenerate $\haH$ is on average upper bounded by
  \begin{equation*}
   {\mathbb E}_{\phi^B_0}  \basisentanglement(\psi^S_0 \otimes \phi^B_0) \leq 2\,\delta\,d_S
  \end{equation*}
  where $\delta = \max_{k} \delta_k$ with $\delta_k = \min_i \tracedistance(\Tr_B \ketbra{E_k}{E_k},\ketbra{i}{i})$ being the geometric measure of entanglement of the eigenstate $\ket{E_k}$ with respect to the basis $\{\ket{i}\}$.
\end{theorem}
\begin{proof}
  Note that $\Tr[\Tr_B\ketbra{E_k}{E_k} \psi^S_0 ]  
  \leq 1 - \tracedistance(\Tr_B\ketbra{E_k}{E_k},\psi^S_0)^2$ and thus all nonzero $|c_k|^2$ in Eq.\ \eqref{eq:definitionofourentaglementmeasure} can be upper bounded by
  \begin{equation*}
    \frac{\Tr[\ketbra{E_k}{E_k} (\psi^S_0 \otimes \phi^B_0 )]}{\Tr[\Tr_B\ketbra{E_k}{E_k} \psi^S_0 ]} (1 - \tracedistance(\Tr_B\ketbra{E_k}{E_k},\psi^S_0)^2) .
  \end{equation*}
  As $(1 - \tracedistance(\Tr_B\ketbra{E_k}{E_k},\psi^S_0)^2) \tracedistance(\Tr_B\ketbra{E_k}{E_k},\psi^S_0) \leq 2\,\delta$, we have
  \begin{equation*}
    \basisentanglement(\psi_0) \leq 2\,\delta \sum_k \frac{\Tr[\ketbra{E_k}{E_k} (\psi^S_0 \otimes \phi^B_0 )]}{\Tr[\Tr_B\ketbra{E_k}{E_k} \psi^S_0 ]} .
  \end{equation*}
  Averaging over all pure states $\phi_0^B$ gives the mean $\id/d_B$ and the sum in the last line is thus upper bounded by $d_S$.
\end{proof}
Note that Theorem~\ref{theo:bounds} implies that whenever a basis $\{\ket{i}\}$ exists for which $\delta$ is small, then for every $i$ there exist many bath states $\phi^B_0$ such that $\basisentanglement(\ketbra{i}{i} \otimes \phi^B_0) \leq 2\,\delta\,d_S$ \cite{Remark2}.
Furthermore, almost all of them will lead to a high effective dimension \cite{Linden09,Gogolin10}.
Obviously the argument can be further strengthened by maximizing $\delta$ only over a subspace that contains most of the probability weight of $\psi_0$: This allows some of the $\delta_k$ to be large, as long as the corresponding $|c_k|^2$ are small.

\begin{figure}[tb]
  \centering%
  \includegraphics[width=\linewidth]{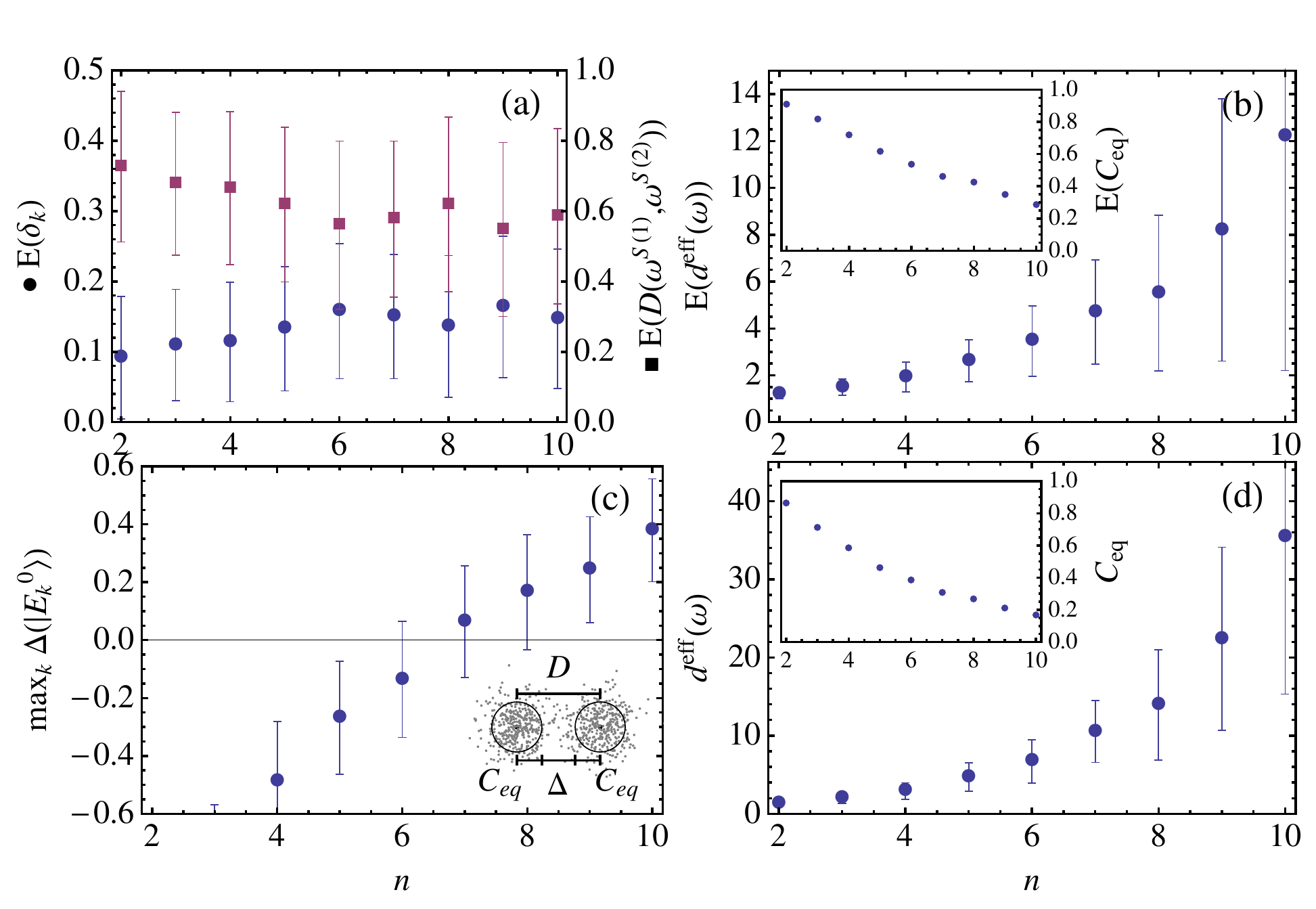}
  \caption{(color online)
    The subsystem is taken to be the first site $S=1$ in the spin chain of $n$ sites, other choices give qualitatively the same results.
    For each of the product eigenvectors $\ket{E^0_k}$ of $\haH_0$ we compare the equilibration properties of $\psi_0^{(1)} = \ket{E^0_k}$ with that of $\psi_0^{(2)} = \sigma^X_S \ket{E^0_k}$ (i.e., the same state but with the first spin flipped) under the dynamics of $\haH$ with $\sigma_1 = 0.4\,\sigma_0$.
    Panels (a), (b) display averages over eigenstates:
    (a)~Average geometric measure of entanglement $\mathbb{E}(\delta_k)$ with respect to the $\haH_0$ eigenbasis and average distance of the reduced time averaged states $\mathbb{E}(D(\omega^{S(1)},\omega^{S(2)}))$. 
    (b)~Average effective dimension and equilibration coefficient [see Eq.~\eqref{eq:equilibrationtheorem}].
    Panels (c), (d) show quantities optimized over eigenstates:
    (c)~Maximum distinguishability $\max_k \Delta(\ket{E_k^{(0)}})$ where $\Delta(\ket{E_k^{(0)}}) = \tracedistance(\omega^{S(1)},\omega^{S(2)}) - \eqcoef(\psi^{(1)}_0) - \eqcoef(\psi^{(2)}_0)$ ($\Delta>0$ ensures distinguishability for almost all times. See the inset for an artists impression.).
    (d)~Effective dimension and equilibration coefficient of the state maximizing $\Delta(\ket{E_k^{(0)}})$.
    All quantities have been averaged over 100 samples from the random Hamiltonian ensemble \eqref{eq:modelhamiltonian}.
    The error bars represent the standard deviation.
    $\deff$ increases rapidly with $n$: hence, equilibration gets better, while the time averaged states on $S$ remain well distinguishable.
    Remarkably the observable that best distinguishes the equilibrium states is $\sigma^Z_S$.
    We find ``equilibration without thermalization'' for a very natural observable.
  } 
  \label{fig:nonintegrableplots}
\end{figure}

%%%%%%%%%%%%%%%%%%%%%%%%%%%%%%%%%%%%%%%%%%

\paragraph{Application to a nonintegrable model.}
The model we consider is a spin-$1/2$ XYZ chain with $n$ sites with random coupling and on-site field.
The Hamiltonian is
\begin{align}
  \label{eq:modelhamiltonian}
  \haH &= \haH_0 + \haH_1 = \sum_{i=1}^n h_i\,\sigma^Z_i + \sum_{i=1}^{n-1} \vec{b}_i \cdot \vec{\sigma}^{\mathrm{NN}}_i ,
\end{align}
where $\vec{\sigma}^{\mathrm{NN}}_i = (\sigma^X_i\sigma^X_{i+1},\sigma^Y_i\sigma^Y_{i+1},\sigma^Z_i\sigma^Z_{i+1})^T$ in terms of the Pauli matrices at site $i$, and $h_i$ and the components of $\vec{b}_i$ are i.i.d.\ normal distributed random variables with zero mean and standard deviations $\sigma_0$ and $\sigma_1$, respectively.
The model is closely related to the one studied in Ref.~\cite{Pal10}.
With unit probability the Hamiltonian is nondegenerate and has nondegenerate gaps. 
We investigate the equilibration properties of the eigenstates of $\haH_0$ after a quench to $\haH$ via exact diagonalization for $\sigma_1 = 0.4\,\sigma_0$.
Hence the integrability breaking term $\haH_1$ is far from being a small perturbation and $\haH$ is nonintegrable in the sense of all of the aforementioned definitions of integrability.
According to the widely accepted belief \cite{kollschollmanm,Rigol08,BanulsCiracHastings10}, one would therefore expect to find thermalization.
However, the numerics suggests that after the quench all local observables \emph{equilibrate}, but \emph{retain memory on the initial conditions} and thus initial state independence is violated.
This conclusion is reached not by keeping track of time evolution, but rather by checking the conditions of Theorem~\ref{theo:mainresult} (see Fig.~\ref{fig:nonintegrableplots}).
It is a challenge to construct nonintegrable models without disorder that violate initial state independence.

\paragraph{Summary and conclusions.}
We have established rigorous results that identified a lack of entanglement in the energy eigenbasis as the reason for an ``equilibration without thermalization'' phenomenon: all local observables equilibrate but retain memory on their initial values for infinitely long time.
By considering a particular model we exemplify that such approximately conserved quantities can exist even in nonintegrable models.
Such models may not saturate Lieb-Robinson bounds, i.e., there probably is no ballistic propagation of information.
Certainly, interesting physical candidates for such models are to be found in disordered systems:
The Anderson model for example has eigenfunctions that are exponentially localized with high probability \cite{Disorder}.
It is the hope that this work stimulates further research on this connection.

\paragraph{Acknowledgments.}
We thank C.\ Neuenhahn and M.\ B.\ Hastings for inspiring discussions. We thank 
the German Academic Society, the EU (Qessence, Compas, Minos), and the 
EURYI for support.

%%%% Bibliography%%%%%%%%%%%%%%%%%%%%%%%%%%%%%%%%%%%%%%%
%  \bibliography{bibliography}

\begin{thebibliography}{34}
\bibitem{CI2}
        The literature is not consistent about whether the conserved quantities have to be local.
\bibitem{Remark2}
        There always exist initial states with identical equilibrium states, 
        thus we only need to show the existence of initial states with identical 
        bath part that equilibrate to different local states.
\end{thebibliography}
%  \bibliographystyle{apsrev}

\section*{EPAPS}

\subsection*{Connection between effective dimension, Loschmidt echo and inverse participation ratio}
The effective dimension of the time averaged state, the time 
average of the Loschmidt echo and the inverse participation ratio of the initial state are closely related in the non-degenerate case.
The inverse participation ratio ($ \mathrm{IPR}$) 
of a state $\psi_0$ is defined to be
\begin{equation*}
  \mathrm{IPR}(\psi_0) = \sum_k |c_k|^4 ,
\end{equation*}
where $c_k = \braket{E_k}{\psi_0}$ are the overlaps with the energy eigenvectors $\{\ket{E_k}\}$ of the Hamiltonian.
The effective dimension is the inverse purity of the time averaged state
\begin{equation*}
  \deff(\omega) = \frac{1}{\Tr[\omega^2]} ,
\end{equation*}
where
\begin{equation*}
  \omega = \overline{\psi_t} = \sum_{k,l} \overline{\ee^{-\iu(E_k - E_l) t}} c_k c_l^* \ketbra{E_k}{E_l}.
\end{equation*}
The Loschmidt echo (see Ref.~\cite{Gritsev10} and the references therein) is defined to be the overlap between the initial state and the 
state at time $t$ 
\begin{align}
  \mathcal{L}_t &= | \braket{\psi_0}{\psi_t} |^2 = |  \bra{\psi_0} \ee^{-\iu\,H\,t} \ket{\psi_0} |^2
  \nonumber\\
  &= \sum_{k,l} \ee^{-\iu(E_k - E_l) t} |c_k|^2 |c_l|^2 .\nonumber
\end{align}
As
\begin{equation*}
  \overline{\ee^{-\iu(E_k - E_l) t}} = \delta_{E_k , E_l}
\end{equation*}
we have in the non-degenerate case that
\begin{equation*}
  \frac{1}{\deff} = \overline{\mathcal{L}_t} = \mathrm{IPR}(\psi_0) = \sum_k |c_k|^4.
\end{equation*}
Due to the ambiguity in the definition of the eigenbasis,
the IPR is not well defined in the degenerate case, and 
\begin{equation*}
  \frac{1}{\deff} = \overline{\mathcal{L}_t} = \sum_{k,l} \delta_{E_k,E_l} |c_k|^2 |c_l|^2 .
\end{equation*}

\end{document}